\documentclass[conference, 9pt]{IEEEtran1}

\usepackage{graphicx}
\usepackage{cite}
\usepackage{amsmath}
\usepackage{amsthm}
\usepackage{subfigure}
\usepackage{float}
\usepackage{times}
\usepackage{enumerate}
\usepackage{amssymb}
\usepackage{algorithm}
\usepackage{algorithmic}
\usepackage{multirow}
\usepackage{braket}
\usepackage{multicol}
\usepackage{caption}
\usepackage{color}

\newtheorem{proposition}{Proposition}[section]

\begin{document}

\title{Network-Level Cooperation in Energy Harvesting Wireless Networks}

\author{\authorblockN{Nikolaos Pappas\textsuperscript{*}, Marios Kountouris\textsuperscript{*}, Jeongho Jeon\textsuperscript{\ddag}, Anthony Ephremides\textsuperscript{\ddag}, Apostolos Traganitis\textsuperscript{\dag}}\\
\authorblockA{\textsuperscript{*} Sup\'{e}lec, Department of Telecommunications, Gif-sur-Yvette, France\\
\textsuperscript{\ddag} Department of Electrical and Computer Engineering and Institute for Systems Research\\
University of Maryland, College Park, MD 20742\\
\textsuperscript{\dag}Computer Science Department, University of Crete, Greece\\Institute of Computer Science, Foundation for Research and Technology - Hellas (FORTH)\\
Email: \{nikolaos.pappas, marios.kountouris\}@supelec.fr, \{jeongho, etony@umd\}.edu, tragani@ics.forth.gr
}

\thanks{This research has been partly supported by the ERC Starting Grant 305123 MORE (Advanced Mathematical Tools for Complex Network Engineering).}}

\maketitle

\begin{abstract}
We consider a two-hop communication network consisted of a source node, a relay and a destination node in which the source and the relay node have external traffic arrivals. The relay forwards a fraction of the source node's traffic to the destination and the cooperation is performed at the network level. In addition, both source and relay nodes have energy harvesting capabilities and an unlimited battery to store the harvested energy. We study the impact of the energy constraints on the stability region. Specifically, we provide inner and outer bounds on the stability region of the two-hop network with energy harvesting source and relay.
\end{abstract}

\IEEEpeerreviewmaketitle

\vspace{-4mm}

\section{Introduction}
\label{sec:ECOOP_intro}

Taking advantage of renewable energy resources from the environment, also known as energy harvesting, allows unattended operability of infrastructure-less wireless networks. There are various forms of energy that can be harvested, including thermal, solar, acoustic, wind, and even ambient radio power. However, the additional functionality of harvesting energy calls for our assessment of the system long-term performance such as in terms of the throughput and stability. In \cite{jeon:stability}, the slotted ALOHA protocol was considered for a network of nodes having energy harvesting capability and the maximum stable throughput region was obtained for bursty traffic.

Cooperative communication helps overcome fading and attenuation in wireless networks. Most cooperative techniques studied so far have been on physical layer cooperation that achieves non-trivial benefits~\cite{b:Yates-NOW}. Nevertheless, there is evidence that the same gains can be achieved with
network layer cooperation, which is plain relaying without any physical layer considerations~\cite{b:Sadek, b:Rong1, b:Pappas-ISIT}.
A key difference between physical layer and network layer cooperation is that the latter can capture the bursty nature of traffic.

In \cite{b:Pappas-JCN}, the authors studied the stability region of a cognitive network under energy constraints. They employed an opportunistic multiple access protocol that observes the priorities among the users to better utilize the limited energy resources. The impact of network-level cooperation in an energy harvesting network with a pure-relay (without its own traffic) under scheduled access was studied in\cite{b:Krikidis_Energy_Harv}.

In this paper, we study the impact of energy constraints on a network with a source-user, a relay and a destination under random access of the medium. Specifically, we provide necessary and sufficient conditions for the stability of a network consisting of a source, a relay and a destination node as shown in Fig.~\ref{fig:ECOOP_model}. We consider the collision channel with erasures and random access of the medium. The source and the relay node have external arrivals; furthermore, the relay is forwarding part of the source node's traffic to the destination, the cooperation is taking place at the network-level.

The analysis is not trivial even for such a simple network because the service process of a node not only depends on the status of its battery but also on the idleness or not of the other node. Note that the reason why the exact region is known only for the two-node and the three-node cases (with or without energy availability constraints) is the \emph{interaction} between the queues of the nodes \cite{tsybakov:ergodicity, rao:stability, Szpankowski:stability}.
First, we obtain an inner bound of the stability region, afterwards we apply the stochastic dominance technique \cite{rao:stability} and Loynes' theorem \cite{b:Loynes} to obtain an outer bound of the stability region.

The rest of this paper is organized as follows. In section~\ref{sec:ECOOP_model}, we define the stability region, describe the channel model, and explain the packet arrival and energy harvesting models. In~\ref{sec:ECOOP_main}, we present an inner and an outer bound of the stability region, the proofs of the results are given in~\ref{sec:ECOOP_analysis}. Finally, we conclude our work in~\ref{sec:ECOOP_conclusion}.

\section{System Model}\label{sec:ECOOP_model}

We consider a time-slotted system in which the nodes randomly access a common receiver as shown in Fig.~\ref{fig:ECOOP_model}, where both source and relay nodes are powered from randomly time-varying renewable energy sources. Each node stores the harvested energy in a battery of unlimited capacity. We denote with $S$, $R$, and $D$ the source, the relay and the destination, respectively. Packet traffic originates from $S$ and $R$. Because of the wireless broadcast nature, $R$ may receive some of the packets transmitted from $S$, which in turn can be relayed to $D$. The packets from $S$ that fail to be received by $D$ but are successfully received by $R$ are relayed by $R$. A half-duplex constraint is imposed here, i.e. $R$ can overhear $S$ only when it is idle.

Each node has an infinite size buffer for storing incoming packets, and the transmission of each packet occupies one time slot. Node $R$ has separate queues for the exogenous arrivals and the endogenous arrivals that are relayed through $R$. Nevertheless, we can let $R$ have a single queue and merge all arrivals into a single queue as the achievable stable throughput region is not affected~\cite{b:Rong3}. This is due to the fact that the link quality between $R$ and $D$ is independent of which packet is selected for transmission.

The packet arrival and energy harvesting processes at $S$ and $R$ are assumed to be Bernoulli distributed with rates $\lambda_S$, $\delta_S$ and $\lambda_R$, $\delta_R$, respectively, and are independent of each other.
$Q_i$ and $B_i$, $i=S,R$, denote the steady state number of packets and energy units in the queue and the energy source at node $i$, respectively. Furthermore, a node $i$ is called active if both its packet queue and its battery are nonempty at the same time, which is denoted by the event $\mathcal{A}_i = \{ B_i \neq 0 \} \cap \{ Q_i \neq 0 \}$ and idle otherwise (denoted by $\overline{\mathcal{A}_i}$).
In each time slot, nodes $S$ and $R$ attempt to transmit with probabilities $q_S$ and $q_R$, respectively, if they are active.
Decisions on transmission are made independently among the nodes and each transmission consumes one energy unit. We assume collision channel with erasures in which if both $S$ and $R$ transmit at the same time slot, a collision occurs and both transmissions fail. The probability that a packet transmitted by node $i$ is successfully decoded at node $j (\neq i)$ is denoted by $p_{ij}$, which is the probability that the signal-to-noise ratio (SNR) over the specified link exceeds a certain threshold for successful decoding. These erasure probabilities capture the effect of random fading at the physical layer. The probabilities $p_{SD}$, $p_{RD}$, and $p_{SR}$ denote the success probabilities over the link $S-D$, $R-D$, and $S-R$, respectively. We also assume that node $R$ has a better channel to $D$ than $S$, i.e. $p_{RD} > p_{SD}$.

The cooperation is performed at the protocol (network) level as follows: when $S$ transmits a packet, if $D$ decodes it successfully, it sends an ACK and the packet exits the network; if $D$ fails to decode the packet but $R$ does, then $R$ sends an ACK and takes over the responsibility of delivering the packet to $D$ by placing it in its queue. If neither $D$ nor $R$ decode (or if $R$ does not store the packet), the packet remains in $S$'s queue for retransmission. The ACKs are assumed to be error-free, instantaneous, and broadcasted to all relevant nodes.

The average service rate for the source node is given by

\begin{equation} \label{eqn:service_rate_S}
\begin{aligned}
\mu_S =  \left\{ q_S (1-q_R) \mathrm{Pr}\left(B_S \neq 0, \mathcal{A}_R \right) +q_S \mathrm{Pr}(B_S \neq 0,\overline{\mathcal{A}_R} ) \right\} \\
\times \left[p_{SD}+(1-p_{SD})p_{SR} \right],
\end{aligned}
\end{equation}
and for the relay is given by

\begin{equation} \label{eqn:service_rate_R}
\mu_R= \left\{q_R (1-q_S) \mathrm{Pr}\left( B_R \neq 0, \mathcal{A}_S \right)+ q_R \mathrm{Pr}(B_R \neq 0,\overline{\mathcal{A}_S} ) \right\} p_{RD}.
\end{equation}

Denote by $Q_i^t$ the length of queue $i$ at the beginning of time slot $t$. Based on the definition in~\cite{Szpankowski:stability}, the queue is said to be \emph{stable} if
\begin{equation*}\label{eqn:PC_definition_stability}
    \lim_{t \rightarrow \infty} {Pr}[Q_i^t < {x}] = F(x)  \text{ and } \lim_{ {x} \rightarrow \infty} F(x) = 1
\end{equation*}
Loynes' theorem~\cite{b:Loynes} states that if the arrival and service processes of a queue are strictly jointly stationary and the average arrival rate is less than the average service rate, then the queue is stable. If the average arrival rate is greater than the average service rate, then the queue is unstable and the value of $Q_i^t$ approaches infinity almost surely. The stability region of the system is defined as the set of arrival rate vectors $\boldsymbol{\lambda}=(\lambda_1, \lambda_2)$ for which the queues in the system are stable.

\begin{figure}[t]
\centering
\includegraphics[scale=0.55]{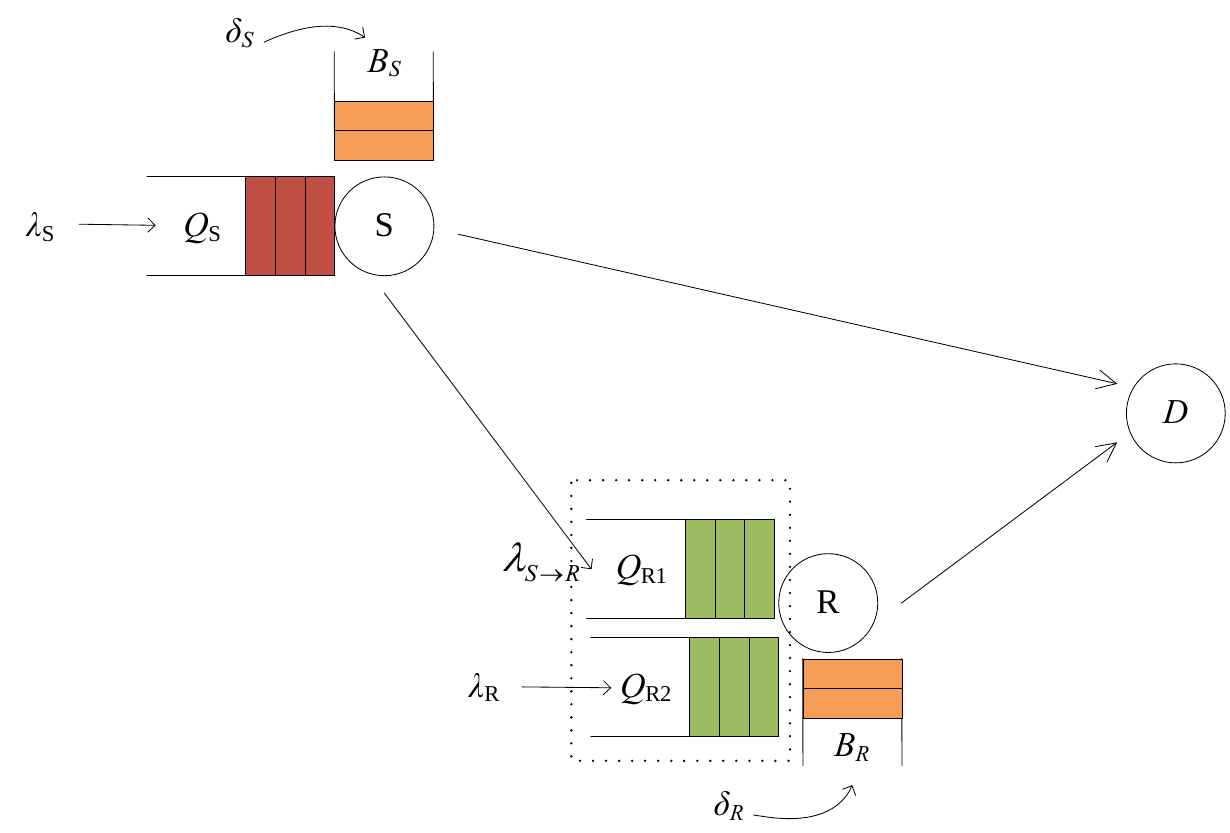}
\caption{The wireless network topology with energy harvesting capabilities.}
\label{fig:ECOOP_model}
\end{figure}

\section{Main Results}\label{sec:ECOOP_main}

This section presents the stability conditions of a network consisting of energy harvesting source and relay, and a destination, as depicted in Fig.~\ref{fig:ECOOP_model}.
The source and the relay are assumed to have infinite size queues to store the harvested energy.

The next proposition presents an inner bound on the stability region by providing sufficient conditions for stability.

\begin{proposition} \label{prop:sufficient_conditions}
If $(\lambda_S, \lambda_R) \in \mathcal{R}_{inner}$, where $\mathcal{R}_{inner}$ is described in (\ref{eqn:R_inner}), then the network in Fig.~\ref{fig:ECOOP_model} is stable.
\end{proposition}

\begin{proof}
The proof is given in Section~\ref{sec:ECOOP_analysis_suff}.
\end{proof}

The following proposition describes an outer bound of the stability region by obtaining necessary conditions for stability.

\begin{proposition} \label{prop:necessary_conditions}
If the network in Fig.~\ref{fig:ECOOP_model} is stable then $(\lambda_S, \lambda_R) \in \mathcal{R}$, where $ \mathcal{R} = \mathcal{R}_1 \bigcup \mathcal{R}_2$. $\mathcal{R}_1$ and $\mathcal{R}_2$ are described by (\ref{eqn:R1}) and (\ref{eqn:R2}) respectively.
\end{proposition}

\begin{proof}
The proof is given in Section~\ref{sec:ECOOP_analysis_necc}.
\end{proof}

\begin{figure*}[!t]
\begin{align} \label{eqn:R_inner}
\mathcal{R}_{inner} = \left\{ (\lambda_{S},\lambda_{R}) :  \lambda_S < \min \left(\delta_S ,q_S \right) \left[1 - \min \left(\delta_R ,q_R \right) \right]  \left[p_{SD}+(1-p_{SD})p_{SR} \right] , \right. \notag \\
\left. \lambda_{R}+ \frac{(1-p_{SD})p_{SR}} {p_{SD}+(1-p_{SD})p_{SR}}\lambda_{S}< \min \left(\delta_R ,q_R \right) \left[1 - \min \left(\delta_S ,q_S \right) \right] p_{RD} \right\}.
\end{align}

\begin{align} \label{eqn:R1}
\mathcal{R}_1 =  \left\{ (\lambda_{S},\lambda_{R}) : \left[1+ \frac{\min(\delta_S,q_S) (1-p_{SD}) p_{SR} }{\left[1-\min(\delta_S,q_S)\right] p_{RD}} \right] \lambda_S + \right. \notag \\
\left.  +\frac{\min(\delta_S,q_S) \left[ p_{SD} + (1-p_{SD})p_{SR} \right]}{\left[1-\min(\delta_S,q_S)\right]p_{RD}} \lambda_R < \min(\delta_S,q_S) \left[ p_{SD} + (1-p_{SD})p_{SR} \right] , \right. \notag \\
\left. \lambda_{R}+ \frac{(1-p_{SD})p_{SR}} {p_{SD}+(1-p_{SD})p_{SR}}\lambda_{S} < \min(\delta_R,q_R) \left[1- \min(\delta_S,q_S)\right] p_{RD} \right\}.
\end{align}

\begin{align} \label{eqn:R2}
\mathcal{R}_2 =  \left\{ (\lambda_{S},\lambda_{R}) : \lambda_R + \frac{\left[1-\min(\delta_R,q_R)\right](1-p_{SD})p_{SR}+\min(\delta_R,q_R) p_{RD}}{\left[1-\min(\delta_R,q_R)\right] \left[p_{SD}+(1-p_{SD})p_{SR} \right]} \lambda_S < \min(\delta_R,q_R) p_{RD} , \right. \notag \\
\left. \lambda_S < \min(\delta_S,q_S) \left[1-\min(\delta_R,q_R)\right] \left[p_{SD}+(1-p_{SD})p_{SR} \right] \right\}
\end{align}
\end{figure*}

\begin{figure}[t]

\centering
 \subfigure[$\mathcal{R}_1$]{
 \includegraphics[scale=0.45]{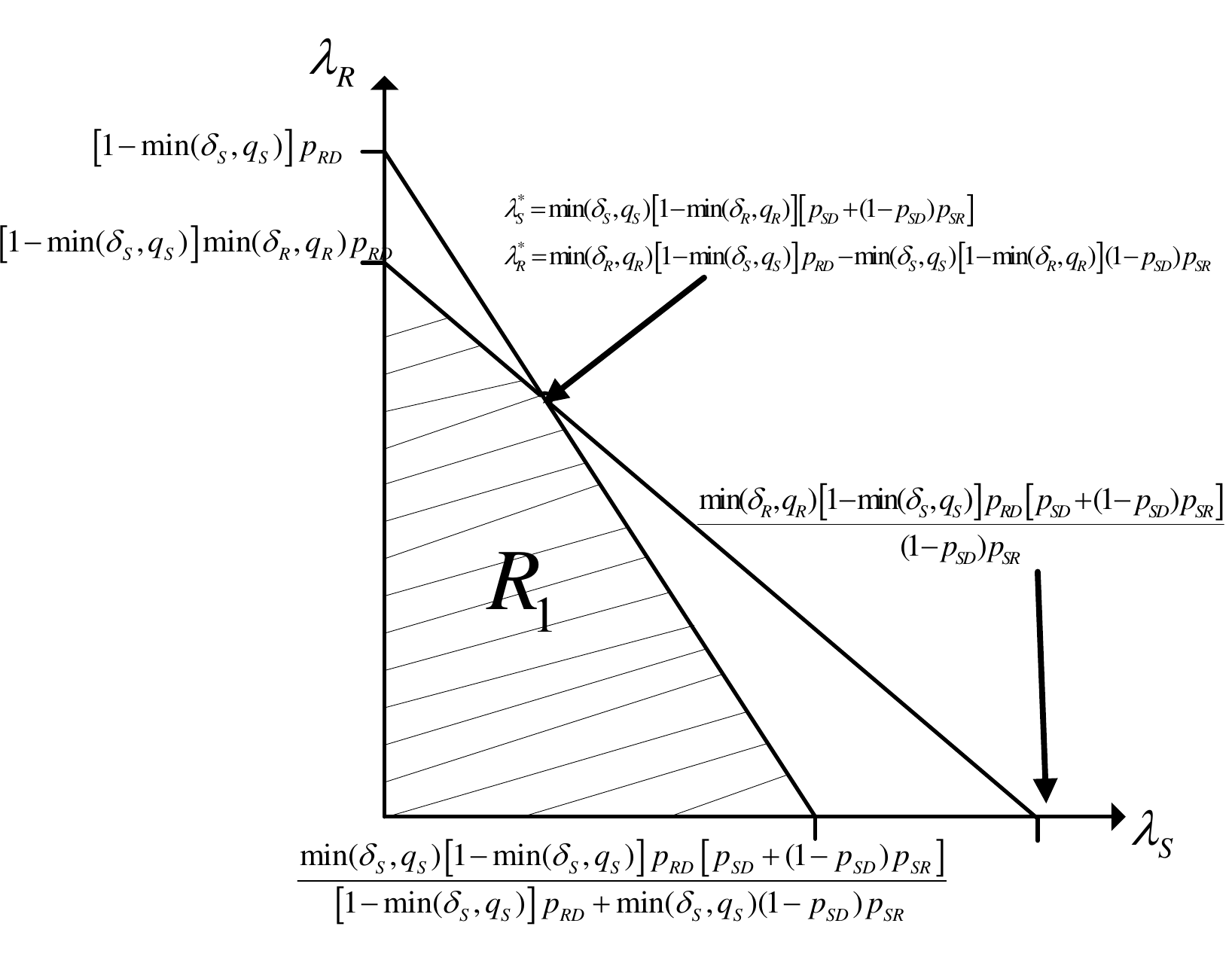}
 \label{fig:ECOOP_1stdom}
 }

  \subfigure[$\mathcal{R}_2$.]{
  \includegraphics[scale=0.45]{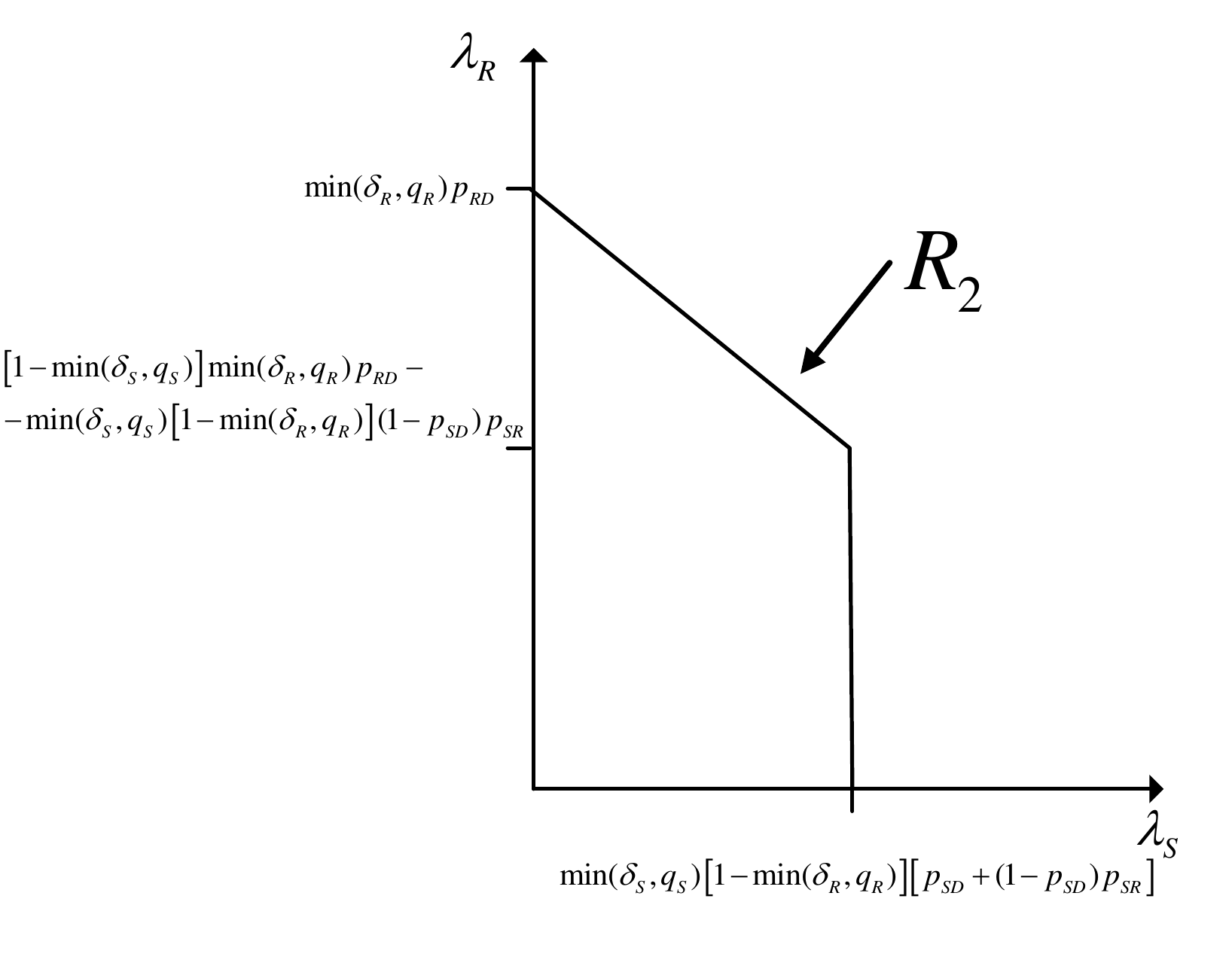}
  \label{fig:ECOOP_2nddom}
  }
  \caption{ An outer bound of the stability region $ \mathcal{R} = \mathcal{R}_1 \bigcup \mathcal{R}_2$,  described in Proposition~\ref{prop:necessary_conditions}.}
\end{figure}

Fig.~\ref{fig:ECOOP_1stdom} and~\ref{fig:ECOOP_2nddom} illustrate the $\mathcal{R}_1$ and $\mathcal{R}_2$ described in Proposition~\ref{prop:necessary_conditions}.

\section{Analysis} \label{sec:ECOOP_analysis}

To derive the stability condition for the queue in the relay node, we need to calculate the total arrival rate. There are two independent arrival processes at the relay: the exogenous traffic with arrival rate $\lambda_R$ and the endogenous traffic from $S$. Denote by $S_A$ the event that $S$ transmits a packet and the packet leaves the queue, then

\begin{equation}
\mathrm{Pr}(S_A)=\left[1-q_{R}\mathrm{Pr}(\mathcal{A}_R)  \right] \left[p_{SD}+(1-p_{SD})p_{SR} \right].
\end{equation}

Among the packets that depart from the queue of $S$, some will exit the network because they are decoded by the destination directly, and some will be relayed by $R$. Denote by $S_B$ the event that the transmitted packet from $S$ will be relayed from $R$, then
\begin{equation}
\mathrm{Pr}(S_B)=\left[1-q_{R}\mathrm{Pr}(\mathcal{A}_R) \right] (1-p_{SD})p_{SR}.
\end{equation}
The conditional probability that a transmitted packet from $S$ is relayed by $R$ given that the transmitted packet exits node $S$'s queue is given by
\begin{equation}
\mathrm{Pr}(S_B | S_A)=\frac{(1-p_{SD})p_{SR}} {p_{SD}+(1-p_{SD})p_{SR}}.
\end{equation}
The arrivals from the source to the relay are
\begin{equation}
\lambda_{S \rightarrow R} = \mathrm{Pr}(S_B | S_A) \lambda_S.
\end{equation}

The total arrival rate at the relay node is given by
\begin{equation}
\lambda_{R,total}=\lambda_{R}+ \frac{(1-p_{SD})p_{SR}} {p_{SD}+(1-p_{SD})p_{SR}}\lambda_{S}.
\end{equation}

\subsection{Sufficient Conditions} \label{sec:ECOOP_analysis_suff}

A queue is considered saturated if in each time slot there is always a packet to transmit, i.e. the queue is never empty.
Assuming saturated queues for the source and the relay node, the saturated throughput for the source node is given by

{\scriptsize
\begin{equation}
\begin{aligned}
\mu_{S}^{s} = \left\{ q_S (1-q_R) \mathrm{Pr}\left(B_S\neq 0, B_R \neq 0 \right) + q_S  \mathrm{Pr} \left(B_S \neq 0, B_R = 0 \right) \right\} \\ \times \left[p_{SD}+(1-p_{SD})p_{SR} \right],
\end{aligned}
\end{equation}
}

and for the relay is given by

\begin{align}
\mu_{R}^{s} = \left\{q_R (1-q_S) \mathrm{Pr}\left(B_S \neq 0, B_R \neq 0 \right)+ \right. \notag \\
\left. + q_R \mathrm{Pr}(B_S = 0 ,B_R \neq 0) \right\} p_{RD}.
\end{align}

Each node transmits with probability $q_i$, $i=S,R$, whenever its battery is not empty and each transmission demands one energy packet. Each energy queue $i$ is then decoupled and forms a discrete-time $M/M/1$ queue with input rate $\delta_i$
and service rate $q_i$, thus the probability the energy queue to be empty is given by
\begin{equation} \label{eqn:prob_nonempty_battery}
\mathrm{Pr}\left(B_i \neq 0 \right) = \min \left( \frac{\delta_i}{q_i} , 1 \right).
\end{equation}

Then, after some calculations we obtain that the saturated throughput for the source is
\begin{equation}
\mu_{S}^{s} = \min \left(\delta_S ,q_S \right) \left[1 - \min \left(\delta_R ,q_R \right) \right]  \left[p_{SD}+(1-p_{SD})p_{SR} \right],
\end{equation}
and for the relay is

\begin{equation} \label{eqn:saturated_thr_R}
\mu_{R}^{s} = \min \left(\delta_R ,q_R \right) \left[1 - \min \left(\delta_S ,q_S \right) \right] p_{RD}.
\end{equation}

The sufficient conditions ($\mathcal{R}_{inner}$) for the stability are obtained by $\lambda_S < \mu_{S}^{s}$ and $\lambda_{R,total} <\mu_{R}^{s}$ and are given by (\ref{eqn:R_inner}), in Proposition~\ref{prop:sufficient_conditions}.

\subsection{Necessary Conditions} \label{sec:ECOOP_analysis_necc}

The average service rates for the source and the relay are given by (\ref{eqn:service_rate_S}) and (\ref{eqn:service_rate_R}), respectively. The average service rate of each queue depends on the status of its own energy and also the
queue size and the energy statues of the other queues. The coupling between the queues (both packet and energy) forms a four dimensional Markov chain which makes the analysis hard.

The stochastic dominant technique~\cite{rao:stability} is essential in order to decouple the interaction between the queues, and thus to characterize the stability region. That is we first construct parallel dominant systems in which one of the nodes transmits dummy packets when its packet queue is empty. Note that even in the dominant system a node cannot transmit if the energy source is empty (because even the dummy packet consumes one energy unit).

We consider the first hypothetical system in which the source node transmits dummy packets when its queue is empty and all the other assumptions remain intact.
The average service rate for the relay given by (\ref{eqn:service_rate_R}) becomes

\begin{equation}
\mu_R= \left\{q_R (1-q_S) \mathrm{Pr}\left(B_S \neq 0, B_R \neq 0 \right)+ q_R \mathrm{Pr}(B_S = 0 ,B_R \neq 0) \right\} p_{RD}.
\end{equation}

The average service rate of the relay, $\mu_{R}$, in the first hypothetical system is the same with the saturated throughput of the relay obtained in (\ref{eqn:saturated_thr_R}). From Loyne's criterion, the relay is stable if $\lambda_{R,total} < \mu_R$.
\begin{equation}
\lambda_{R}+ \frac{(1-p_{SD})p_{SR}} {p_{SD}+(1-p_{SD})p_{SR}}\lambda_{S} < \min \left(\delta_R ,q_R \right) \left[1 - \min \left(\delta_S ,q_S \right) \right] p_{RD}.
\end{equation}

The average number of packets per active slot for $R$ is $\left[1 - \min \left(\delta_S ,q_S \right) \right] q_R p_{RD}$, thus the fraction of active slots is given by
\begin{equation} \label{eqn:prob_active_R_1st_hyp_system}
 \mathrm{Pr}\left(B_R \neq 0, Q_R\neq 0 \right) = \frac{\lambda_{R}+ \frac{(1-p_{SD})p_{SR}} {p_{SD}+(1-p_{SD})p_{SR}}\lambda_{S}}{\left[1-\min \left(\delta_S ,q_S \right) \right] q_R p_{RD}}
\end{equation}

After changing (\ref{eqn:prob_nonempty_battery}) and (\ref{eqn:prob_active_R_1st_hyp_system}) into (\ref{eqn:service_rate_S}), the service rate for the source becomes

\begin{equation}
\begin{aligned}
\mu_S= \min \left(\delta_S ,q_S \right) \left[1 - \frac{\lambda_{R}+ \frac{(1-p_{SD})p_{SR}} {p_{SD}+(1-p_{SD})p_{SR}}\lambda_{S}}{\left[1-\min \left(\delta_S ,q_S \right) \right] p_{RD}} \right] \\
\times \left[p_{SD}+(1-p_{SD})p_{SR} \right].
\end{aligned}
\end{equation}

The queue in $S$ is stable if $\lambda_S < \mu_S$ and after some manipulations we obtain

\begin{equation}
\begin{aligned}
\left[1+ \frac{\min \left(\delta_S ,q_S \right) (1-p_{SD}) p_{SR} }{\left[1-\min \left(\delta_S ,q_S \right) \right]p_{RD}} \right] \lambda_S + \\
+\frac{\min \left(\delta_S ,q_S \right) \left[ p_{SD} + (1-p_{SD})p_{SR} \right]}{
\left[1-\min \left(\delta_S ,q_S \right) \right]p_{RD}} \lambda_R \\
 < \min \left(\delta_S ,q_S \right) \left[ p_{SD} + (1-p_{SD})p_{SR} \right].
\end{aligned}
\end{equation}

The derived stability conditions from the first hypothetical system are summarized in (\ref{eqn:R1}).

In the second hypothetical system, the relay node transmits dummy packets and all the other assumptions remain intact. Thus, the average service rate for the source given by (\ref{eqn:service_rate_S}), becomes

\begin{equation}
\begin{aligned}
\mu_S= \left\{ q_S (1-q_R) \mathrm{Pr}\left(B_S \neq 0, B_R \neq 0 \right) +q_S \mathrm{Pr}(B_S \neq 0, B_R = 0) \right\} \\
\times \left[p_{SD}+(1-p_{SD})p_{SR} \right],
\end{aligned}
\end{equation}

which is equal to saturated throughput of the source and is given by

\begin{equation}
\begin{aligned}
\mu_S=\min \left(\delta_S ,q_S \right) \left[1 - \min \left(\delta_R ,q_R \right) \right]  \left[p_{SD}+(1-p_{SD})p_{SR} \right].
\end{aligned}
\end{equation}
From Loyne's theorem, the queue in source is stable if $\lambda_S < \mu_S$ thus
\begin{equation}
\begin{aligned}
\lambda_S < \min \left(\delta_S ,q_S \right) \left[1 - \min \left(\delta_R ,q_R \right) \right] \left[p_{SD}+(1-p_{SD})p_{SR} \right].
\end{aligned}
\end{equation}
The average number of packets per active slot for $S$ is $ q_S \left[1 - \min \left(\delta_R ,q_R \right) \right] \left[p_{SD}+(1-p_{SD})p_{SR} \right] $.
The fraction of active slots for the source $S$ is
\begin{equation} \label{eqn:prob_active_S_2nd_hyp_system}
 \mathrm{Pr}\left(B_S \neq 0, Q_S\neq 0 \right) = \frac{\lambda_{S}}{q_S \left[1 - \min \left(\delta_R ,q_R \right) \right] \left[p_{SD}+(1-p_{SD})p_{SR} \right]}.
\end{equation}

After replacing from (\ref{eqn:prob_nonempty_battery}) and (\ref{eqn:prob_active_S_2nd_hyp_system}) into (\ref{eqn:service_rate_R}), the service rate for the relay is

{\scriptsize
\begin{equation}
\mu_R=\min \left(\delta_R ,q_R \right) \left[1- \frac{\lambda_{S}}{\left[1 - \min \left(\delta_R ,q_R \right) \right] \left[p_{SD}+(1-p_{SD})p_{SR} \right]}\right] p_{RD}.
\end{equation}
}

The queue in the relay node $R$ is stable if $\lambda_{R,total} < \mu_R$ and after some manipulations we obtain

\begin{equation}
\begin{aligned}
\lambda_R + \frac{\left[1 - \min \left(\delta_R ,q_R \right) \right](1-p_{SD})p_{SR}+\min \left(\delta_R ,q_R \right) p_{RD}}{\left[1 - \min \left(\delta_R ,q_R \right) \right] \left[p_{SD}+(1-p_{SD})p_{SR} \right]} \lambda_S  \\
< \min \left(\delta_R ,q_R \right) p_{RD}.
\end{aligned}
\end{equation}

The derived stability conditions from the second hypothetical system are given by (\ref{eqn:R2}).

An important observation made in \cite{rao:stability} is that the stability conditions obtained by using the stochastic dominance technique are not merely sufficient conditions for the stability of the original system but are sufficient and necessary conditions. However, the \emph{indistinguishability} argument does not apply to our problem. In a system with batteries, the dummy packet transmissions affect the dynamics of the batteries. For example, there are instants when a node is no more able to transmit in the hypothetical system
because of the lack of energy, while it is able to transmit in the original system, thus it may result to a better chance of success for the other node.

The obtained stability conditions are necessary conditions of the original system and are summarized in Proposition~\ref{prop:necessary_conditions}.

\section{Conclusion} \label{sec:ECOOP_conclusion}

In this paper, we studied the effect of energy constraints on a wireless network with energy harvesting source and relay and a destination. The source and the relay nodes have external arrivals and network-level cooperation is employed in which the relay forwards a fraction of the source's traffic to the destination. We derived necessary and sufficient stability conditions of the above cooperative communication scenario. A next step is to obtain the closure for the inner and outer bounds presented here. Further extensions will investigate the effect of finite battery capacity and that of multi-packet reception.

\vspace{-5mm}
\bibliographystyle{IEEEtran}
\bibliography{thesis}

\end{document}